\pdfoutput=1
\RequirePackage{ifpdf}
\ifpdf % We~are running pdfTeX in pdf mode
\documentclass[pdftex]{sigma}
\else
\documentclass{sigma}
\fi

\numberwithin{equation}{section}

\newtheorem{Theorem}{Theorem}[section]
\newtheorem{Proposition}[Theorem]{Proposition}
 { \theoremstyle{definition}
\newtheorem{Remark}[Theorem]{Remark} }

\begin{document}
	
\allowdisplaybreaks

\newcommand{\arXivNumber}{2403.05464}

\renewcommand{\PaperNumber}{049}

\FirstPageHeading

\ShortArticleName{Generalized Yang Poisson Models on Canonical Phase Space}

\ArticleName{Generalized Yang Poisson Models\\ on Canonical Phase Space}	

\Author{Tea MARTINI\'{C} BILA\'{C}~$^{\rm a}$, Stjepan MELJANAC~$^{\rm b}$ and Salvatore MIGNEMI~$^{\rm cd}$}

\AuthorNameForHeading{T.~Martini\'{c} Bila\'{c}, S.~Meljanac and S.~Mignemi}

\Address{$^{\rm a)}$~Faculty of Science, University of Split, Rudera Bo\v{s}kovi\'ca 33, 21000 Split, Croatia}
\EmailD{\href{mailto:rina.teamar@pmfst.hr}{teamar@pmfst.hr}}

\Address{$^{\rm b)}$~Division of Theoretical Physics, Ruder Bo\v{s}kovi\'c Institute,\\
\hphantom{$^{\rm b)}$}~Bijeni\v{c}ka cesta 54, 10002 Zagreb, Croatia}
\EmailD{\href{mailto:meljanac@irb.hr}{meljanac@irb.hr}}

 \Address{$^{\rm c)}$~Dipartimento di Matematica, Universit\`a di Cagliari, via Ospedale 72, 09124 Cagliari, Italy}
\EmailD{\href{mailto:smignemi@unica.it}{smignemi@unica.it}}

\Address{$^{\rm d)}$~INFN, Sezione di Cagliari Cittadella Universitaria, 09042 Monserrato, Italy}

\ArticleDates{Received March 12, 2024, in final form May 29, 2024; Published online June 10, 2024}

\Abstract{We discuss the generalized Yang Poisson models. We construct generalizations of the Yang Poisson algebra related to $\mathfrak{o}(1,5)$ algebra discussed by Meljanac and Mignemi~(2023). The exact realizations of this generalized algebra on canonical phase space are presented and the corresponding differential equations are solved in simple cases. Furthermore, we discuss the Poisson algebras related to $\mathfrak{o}(3,3)$ and $\mathfrak{o}(2,4)$ algebras.}	
	

\Classification{81R60}

\section{Yang Poisson model}\label{sec-1}

The Yang model of noncommutative geometry on a curved background spacetime was proposed in \cite{yang} and is a generalization of the noncommutative geometry first introduced by Snyder \cite{Snyder}. This model is defined in terms of a 15-parameter algebra isomorphic to $\mathfrak{o}(1, 5)$ containing the generators of the Lorentz algebra with the coordinates of phase spaces, and was generalized in~\cite{Khruschev}. Later, Snyder--de Sitter or triply special relativity model, where this symmetry is realized nonlinearly, was proposed in \cite{Banerjee, Guo, Kowalski, Mignemi-1}.
Realizations of the Yang model have been discussed in \cite{Lukierski, Martinic-1, Meljanac-1, Meljanac-2, Meljanac-3}. They cannot be obtained in closed analytic form, but only as power series in coupling constants. In the classical limit, where quantum commutators are replaced by Poisson brackets, the algebra has a simpler structure and no ordering problems arise. The corresponding realizations of this Yang Poisson model can be obtained in an exact form \cite{Meljanac-4}.

The Yang Poisson model is generated by a Poisson algebra containing the usual Lorentz algebra of generators $M_{\mu\nu}$, with its standard action on phase space parametrized by $\hat{x}_{\mu}$ and $\hat{p}_{\mu}$,
\begin{gather}\label{0.1}
\bigl\lbrace M_{\mu\nu},M_{\rho\sigma}\bigr\rbrace =\bigl( \eta_{\mu\rho}M_{\nu\sigma}-\eta_{\mu\sigma}M_{\nu\rho}-\eta_{\nu\rho}M_{\mu\sigma}+\eta_{\nu\sigma}M_{\mu\rho}\bigr),
\\ \label{0.2}
\bigl\lbrace M_{\mu\nu}, \hat{x}_{\lambda}\bigr\rbrace =\bigl( \eta_{\mu\lambda}\hat{x}_{\nu}-\eta_{\nu\lambda}\hat{x}_{\mu}\bigr) ,\qquad \bigl\lbrace M_{\mu\nu}, \hat{p}_{\lambda}\bigr\rbrace =\bigl( \eta_{\mu\lambda}\hat{p}_{\nu}-\eta_{\nu\lambda}\hat{p}_{\mu}\bigr),
\\ \label{0.3}
\bigl\lbrace \hat{x}_{\mu},\hat{x}_{\nu}\bigr\rbrace =\beta^{2}M_{\mu\nu}, \qquad \bigl\lbrace \hat{p}_{\mu},\hat{p}_{\nu}\bigr\rbrace =\alpha^{2}M_{\mu\nu},
\\ \label{0.4}
\bigl\lbrace \hat{x}_{\mu}, \hat{p}_{\nu}\bigr\rbrace =\eta_{\mu\nu}h,
\\ \label{1.7}
\bigl\lbrace h, \hat{x}_{\mu}\bigr\rbrace =\beta^{2}\hat{p}_{\mu}, \qquad \bigl\lbrace h, \hat{p}_{\mu}\bigr\rbrace =-\alpha^{2}\hat{x}_{\mu},
\\ \label{0.5}
\bigl\lbrace M_{\mu\nu}, h \bigr\rbrace =0,
\end{gather}
where $\lbrace ,\rbrace $ are classical Poisson brackets, $\mu,\nu=0,\ldots, n-1$, $ \eta_{\mu\nu}={\rm diag}(-1,1,\ldots,1)$ % RS: diag -> {\rm diag}
is the Minkowski metric and $\alpha$ and $\beta$ are real parameters.
Generators denoted by $A$, $B$, $C$ of this Poisson algebra satisfy Jacobi relations of type
\begin{equation*}
 \lbrace A, \lbrace B, C \rbrace \rbrace + \lbrace B, \lbrace C, A \rbrace \rbrace+ \lbrace C, \lbrace A, B \rbrace \rbrace=0.
\end{equation*}
In the limit when $\alpha \rightarrow 0$ the algebra \eqref{0.1}--\eqref{0.5} becomes the Snyder Poisson algebra and in the limit when $\beta \rightarrow 0$, it becomes the dual Snyder Poisson algebra related to the de Sitter algebra~$\mathfrak{o}(1,5),$ for $n=4, $ and $ \mu,\nu=0,1,2,3$.

We look for realizations of $\hat{x}_{\mu}$ and $\hat{p}_{\mu}$ on a phase space with coordinates $x_{\mu}$ and momenta~$p_{\mu}$ satisfying the canonical algebra
\begin{equation*}
\bigl\lbrace x_{\mu},x_{\nu}\bigr\rbrace =\bigl\lbrace p_{\mu},p_{\nu}\bigr\rbrace =0, \qquad \bigl\lbrace x_{\mu},p_{\nu}\bigr\rbrace =\eta_{\mu\nu}.
\end{equation*}
Realizations of the Yang Poisson algebra on a canonical phase space of coordinates $x_{\mu}$ and $p_{\mu}$, with the Lorentz algebra generators realized as $M_{\mu\nu}=x_{\mu}p_{\nu}-x_{\nu}p_{\mu}$ were analyzed and discussed in \cite{Meljanac-4}.
Special solution for realizations of $\hat{x}_{\mu}$ and $\hat{p}_{\mu}$ are given by
\begin{gather}\label{1}
\hat{x}_{\mu}=x_{\mu}\sqrt{1-\beta^{2}p^{2}+\phi_{1}(z)},
\\ \label{1.1}
\hat{p}_{\mu}=p_{\mu}\sqrt{1-\alpha^{2}x^{2}+\phi_{2}(z)},
\end{gather}
where
\begin{gather}
\phi_{1}\phi_{2}+\phi_{1}+\phi_{2}=z^{2}, \qquad z=\alpha\beta ( xp),\nonumber\\
\label{1.5}
h=\sqrt{\bigl( 1-\beta^{2}p^{2}+\phi_{1}(z)\bigr) \bigl( 1-\alpha^{2}x^{2}+\phi_{2}(z)\bigr) }.
\end{gather}
Generally, it holds that in terms of original variables $\hat x_{\mu}$ and $\hat p_{\mu}$,
\begin{equation}\label{1.6}
h=\sqrt{1-\alpha^{2}\hat{x}^{2}-\beta^{2}\hat{p}^{2}-\frac{\alpha^{2}\beta^{2}}{2}M^{2}}.
\end{equation}
This result for $h$ is universal and generally valid in the Yang Poisson model \cite{Meljanac-4}. The most general realizations are obtained using the group of automorphisms applied to the special solution defined with $\phi_{2}(z)=0$ \cite{Meljanac-4}.

In the present paper, we shall consider generalized Yang Poisson models and their related models, as well as their exact realizations.
In Section \ref{sec-2}, we define generalizations of the Yang Poisson algebra \eqref{0.1}--\eqref{0.5} and construct realizations of this algebra solving the corresponding differential equations, presenting the general results in Section \ref{sec-3}. Furthermore, we discuss the related Poisson algebras in Section \ref{sec-4}.

\section{Generalized Yang Poisson models}\label{sec-2}

The most general new generators linear in $\hat{x}_{\mu}$, $\hat{p}_{\mu}$ and $M_{\mu\nu}$ are
\begin{gather}\label{x-2}
\tilde{X}_{\mu}=A\biggl( \cos \varphi \hat{x}_{\mu}+\frac{\beta}{\alpha}\sin \varphi \hat{p}_{\mu}\biggr) +\beta a_{\nu}M_{\mu\nu},
\\ \label{p-2}
\tilde{P}_{\mu}=B\biggl( \cos \psi \hat{p}_{\mu}+\frac{\alpha}{\beta}\sin \psi \hat{x}_{\mu}\biggr) +\alpha b_{\nu}M_{\mu\nu},
\end{gather}
where $M_{\mu\nu}=x_{\mu}p_{\nu}-x_{\nu}p_{\mu}$ is unchanged and $A$, $B$, $\varphi$, $\psi$, $a_{\mu}$, $b_{\mu}$ are dimensionless parameters with~$AB\neq 0$.

The inverse transformations are
\begin{gather*}
\hat{x}_{\mu}=\frac{A^{-1}\alpha\cos \psi\bigl(\tilde{X}_{\mu}-\beta a_{\nu}M_{\mu\nu}\bigr) -B^{-1}\beta\sin \varphi\bigl(\tilde{P}_{\mu}-\alpha b_{\nu}M_{\mu\nu}\bigr) }{\alpha \cos ( \varphi +\psi) },
\\
\hat{p}_{\mu}=\frac{B^{-1}\beta\cos \varphi\bigl(\tilde{P}_{\mu}-\alpha b_{\nu}M_{\mu\nu}\bigr)-A^{-1}\alpha\sin \psi\bigl(\tilde{X}_{\mu}-\beta a_{\nu}M_{\mu\nu}\bigr)}{\beta \cos ( \varphi +\psi)}
.\end{gather*}
The new generators $\tilde{X}_{\mu}$ and $\tilde{P}_{\mu}$ generate a new class of Poisson algebras isomorphic to the initial Yang Poisson algebra. The new algebra generated by $\tilde{X}_{\mu}$, $\tilde{P}_{\mu}$, $M_{\mu\nu}$ and $\tilde{h}$ is given by
\begin{gather}\label{1.10}
\bigl\lbrace \tilde{X}_{\mu}, \tilde{X}_{\nu}\bigr\rbrace =\bigl( \beta^{2}\tilde{A}M_{\mu\nu}+\beta\bigl( a_{\mu}\tilde{X}_{\nu}-a_{\nu}\tilde{X}_{\mu} \bigr)\bigr) ,
\\ \label{1.11}
\bigl\lbrace \tilde{P}_{\mu}, \tilde{P}_{\nu}\bigr\rbrace =\bigl( \alpha^{2}\tilde{B}M_{\mu\nu}+\alpha\bigl( b_{\mu}\tilde{P}_{\nu}-b_{\nu}\tilde{P}_{\mu}\bigr)\bigr) ,
\\ \label{1.12}
\bigl\lbrace \tilde{X}_{\mu}, \tilde{P}_{\nu}\bigr\rbrace=\bigl( \eta_{\mu\nu}\tilde{h}+\alpha b_{\mu}\tilde{X}_{\nu}-\beta a_{\nu}\tilde{P}_{\mu}+\alpha\beta AB\tilde{\rho}M_{\mu\nu}\bigr),
\\ \label{1.13}
\bigl\lbrace M_{\mu\nu},\tilde{X}_{\lambda}\bigr\rbrace =\bigl( \eta_{\mu\lambda}\tilde{X}_{\nu}-\eta_{\nu\lambda}\tilde{X}_{\mu}+\beta\bigl( a_{\mu}M_{\lambda\nu}-a_{\nu}M_{\lambda\mu}\bigr) \bigr),
\\ \label{1.14}
\bigl\lbrace M_{\mu\nu},\tilde{P}_{\lambda}\bigr\rbrace =\bigl( \eta_{\mu\lambda}\tilde{P}_{\nu}-\eta_{\nu\lambda}\tilde{P}_{\mu}+\alpha\bigl( b_{\mu}M_{\lambda\nu}-b_{\nu}M_{\lambda\mu}\bigr) \bigr),
\\ \label{1.15}
\bigl\lbrace M_{\mu\nu},\tilde{h}\bigr\rbrace =\alpha\bigl( b_{\nu}\tilde{X}_{\mu}-b_{\mu}\tilde{X}_{\nu}\bigr) -\beta\bigl( a_{\nu}\tilde{P}_{\mu}-a_{\mu}\tilde{P}_{\nu}\bigr),
\\ \label{2}
\bigl\lbrace \tilde{h}, \tilde{X}_{\mu}\bigr\rbrace =\beta^{2}\tilde{A}\tilde{P}_{\mu}-\alpha\beta AB\tilde{\rho}\tilde{X}_{\mu}-\beta a_{\mu}\tilde{h},
\\ \label{2.1}
\bigl\lbrace \tilde{h}, \tilde{P}_{\mu}\bigr\rbrace=-\bigl( \alpha^{2}\tilde{B}\tilde{X}_{\mu}-\alpha\beta AB\tilde{\rho}\tilde{P}_{\mu}-\alpha b_{\mu}\tilde{h}\bigr),
\end{gather}
where
\begin{gather}\label{1.16}
\tilde{h}=AB\cos(\varphi+\psi)h+\beta a\tilde{P}-\alpha b \tilde{X}-\alpha\beta a_{\mu}b_{\nu}M_{\mu\nu},
\\
\label{1.17}
\tilde{\rho}=\sin(\varphi +\psi)+\frac{ab}{AB}, \qquad \tilde{A}=A^{2}+a^{2} \qquad \text{and} \qquad \tilde{B}=B^{2}+b^{2}.
\end{gather}
The algebra \eqref{1.10}--\eqref{1.17} is invariant under Born duality \cite{Born},
$\alpha \leftrightarrow \beta$, $a_{i}\rightarrow -b_{i}$, $b_{i}\rightarrow a_{i}$, $\tilde{A} \leftrightarrow \tilde{B}$, $\tilde{\rho} \leftrightarrow -\tilde{\rho}$, $\tilde{X}_{\mu} \rightarrow -\tilde{P}_{\mu}$, $\tilde{P}_{\mu} \rightarrow \tilde{X}_{\mu}$, $M_{\mu\nu}\leftrightarrow M_{\mu\nu}$, $\tilde{h} \leftrightarrow \tilde{h}$.

In the following, we consider the case where $a_{\mu}=b_{\mu}=0$.
In this case, the Poisson algebra~\mbox{\eqref{1.10}--\eqref{1.17}} corresponds to the classical limit of the Khruschev--Leznov algebra \cite{Khruschev}, where quantum commutators are replaced by Poisson brackets. Let us denote $\tilde{\alpha}^{2}=\alpha^{2}\tilde{B}$ and $\tilde{\beta}^{2}=\beta^{2}\tilde{A}$. Then we have
\begin{equation*}
\tilde{h}=AB\cos(\varphi+\psi)\sqrt{1-\tilde{\alpha}^{2}\tilde{X}^{2}-
	\tilde{\beta}^{2}\tilde{P}^{2}+2\tilde{\rho}\tilde{\alpha}\tilde{\beta}\tilde{X}\tilde{P}-\frac{\tilde{\alpha}^{2}\tilde{\beta}^{2}}{2A^{2}B^{2}}M^{2}},
\end{equation*}
where
\begin{equation*}
\tilde{\rho}=\sin(\varphi +\psi), \qquad \tilde{A}=A^{2} \qquad \text{and} \qquad \tilde{B}=B^{2}.
\end{equation*}
Using the realizations \eqref{1} and \eqref{1.1} of $\hat{x}_{\mu}$ and $\hat{p}_{\mu}$, we can write realization of $\tilde{X}_{\mu}$, $\tilde{P}_{\mu}$ in terms of $x_{\mu}$ and $p_{\mu}$,
\begin{gather}\label{1.3}
\tilde{X}_{\mu}=x_{\mu}A\cos(\varphi)\sqrt{1-\frac{\tilde{\beta}^{2}}{A^{2}}p^{2}+\varphi_{1}(z)}+\frac{\tilde{\beta}}{\tilde{\alpha}}p_{\mu}B\sin(\varphi)\sqrt{1-\frac{\tilde{\alpha}^{2}}{B^{2}}x^{2}+\varphi_{2}(z)},
\\ \label{1.4}
\tilde{P}_{\mu}=p_{\mu}B\cos(\psi)\sqrt{1-\frac{\tilde{\alpha}^{2}}{B^{2}}x^{2}+\varphi_{2}(z)}+\frac{\tilde{\alpha}}{\tilde{\beta}}x_{\mu}A\sin(\varphi)\sqrt{1-\frac{\tilde{\beta}^{2}}{A^{2}}p^{2}+\varphi_{1}(z)}
.\end{gather}

Generalizing \eqref{1.3} and \eqref{1.4}, we get the general ansatz for $\tilde{X}_{\mu}$ and $\tilde{P}_{\mu}$
\begin{equation}\label{2.18}
\tilde{X}_{\mu}=x_{\mu}f+\frac{\tilde{\beta}}{\tilde{\alpha}}p_{\mu}g
\end{equation}
and
\begin{equation}\label{2.19}
\tilde{P}_{\mu}=p_{\mu}\tilde{f}+\frac{\tilde{\alpha}}{\tilde{\beta}}x_{\mu}\tilde{g},
\end{equation}
where $f$, $g$, $\tilde{f}$, $\tilde{g}$ are functions of $u$, $v$, $z$, with $u=\tilde{\beta}^{2}p^{2}$, $v=\tilde{\alpha}^{2}x^{2}$, $z=\frac{\tilde{\alpha}\tilde{\beta}}{AB}( xp)$.

From $\bigl\lbrace \tilde{X}_{\mu}, \tilde{X}_{\nu}\bigr\rbrace =\tilde{ \beta^{2}}M_{\mu\nu}$, we get
\begin{align}
&-2f\frac{\partial f}{\partial u}-2g\frac{\partial g}{\partial v}+4z\biggl( \frac{\partial f}{\partial v}\frac{\partial g}{\partial u}-
\frac{\partial f}{\partial u}\frac{\partial g}{\partial v}\biggr) +2v\biggl( \frac{\partial f}{\partial v}\frac{\partial g}{\partial z}-\frac{\partial f}{\partial z}\frac{\partial g}{\partial v}\biggr) \notag\\
&\qquad{}+2u\biggl( \frac{\partial f}{\partial z}\frac{\partial g}{\partial u}-\frac{\partial f}{\partial u}\frac{\partial g}{\partial z}\biggr)+f\frac{\partial g}{\partial z}+g\frac{\partial f}{\partial z}
=1 \label{e-1}
\end{align}
and from $\bigl\lbrace \tilde{P}_{\mu}, \tilde{P}_{\nu}\bigr\rbrace =\tilde{\alpha}^{2}M_{\mu\nu}$ it follows
\begin{align}
&-2\tilde{f}\frac{\partial \tilde{f}}{\partial v}-2\tilde{g}\frac{\partial \tilde{g}}{\partial u}+4z\biggl( \frac{\partial \tilde{f}}{\partial u}\frac{\partial \tilde{g}}{\partial v}-
\frac{\partial \tilde{f}}{\partial v}\frac{\partial \tilde{g}}{\partial u}\biggr) +2v\biggl( \frac{\partial \tilde{f}}{\partial z}\frac{\partial \tilde{g}}{\partial v}-\frac{\partial \tilde{f}}{\partial v}\frac{\partial \tilde{g}}{\partial z}\biggr) \notag\\
&\qquad{}+2u\biggl( \frac{\partial \tilde{f}}{\partial u}\frac{\partial \tilde{g}}{\partial z}-\frac{\partial \tilde{f}}{\partial z}\frac{\partial \tilde{g}}{\partial u}\biggr)+\tilde{f}\frac{\partial \tilde{g}}{\partial z}+\tilde{g}\frac{\partial \tilde{f}}{\partial z}
=1.\label{e-2}
\end{align}
The relation
$\bigl\lbrace \tilde{X}_{\mu}, \tilde{P}_{\nu}\bigr\rbrace= \eta_{\mu\nu}\tilde{h}+\tilde{\alpha}\tilde{\beta} \tilde{\rho}M_{\mu\nu}$ yields following five equations:
\begin{align} \label{e-3}
&f\tilde{f}-g\tilde{g}=\tilde{h},
\\
&2\tilde{f}\frac{\partial f}{\partial v}-2g\frac{\partial \tilde{g}}{\partial v}+4z\biggl( \frac{\partial f}{\partial v}\frac{\partial \tilde{g}}{\partial u}-
\frac{\partial f}{\partial u}\frac{\partial \tilde{g}}{\partial v}\biggr) +2v\biggl( \frac{\partial f}{\partial v}\frac{\partial \tilde{g}}{\partial z}-\frac{\partial f}{\partial z}\frac{\partial \tilde{g}}{\partial v}\biggr) \notag \\
&\qquad{}+2u\biggl( \frac{\partial f}{\partial z}\frac{\partial \tilde{g}}{\partial u}-\frac{\partial f}{\partial u}\frac{\partial \tilde{g}}{\partial z}\biggr)+f\frac{\partial \tilde{g}}{\partial z}-\tilde{g}\frac{\partial f}{\partial z}
=0,\label{e-4}
\\
&-2f\frac{\partial \tilde{f}}{\partial u}-2\tilde{g}\frac{\partial g}{\partial u}+4z\biggl( \frac{\partial \tilde{f}}{\partial u}\frac{\partial g}{\partial v}-
\frac{\partial \tilde{f}}{\partial v}\frac{\partial g}{\partial u}\biggr) +2v\biggl( \frac{\partial \tilde{f}}{\partial z}\frac{\partial g}{\partial v}-\frac{\partial \tilde{f}}{\partial v}\frac{\partial g}{\partial z}\biggr) \notag \\
&\qquad{}+2u\biggl( \frac{\partial \tilde{f}}{\partial u}\frac{\partial g}{\partial z}-\frac{\partial \tilde{f}}{\partial z}\frac{\partial g}{\partial u}\biggr)+\tilde{f}\frac{\partial g}{\partial z}-g\frac{\partial \tilde{f}}{\partial z}=0,\label{e-5}
\\
&-2\tilde{g}\frac{\partial f}{\partial u}-2g\frac{\partial \tilde{f}}{\partial v}+4z\biggl( \frac{\partial f}{\partial v}\frac{\partial \tilde{f}}{\partial u}-
\frac{\partial f}{\partial u}\frac{\partial \tilde{f}}{\partial v}\biggr) +2v\biggl( \frac{\partial f}{\partial v}\frac{\partial \tilde{f}}{\partial z}-\frac{\partial f}{\partial z}\frac{\partial \tilde{f}}{\partial v}\biggr)\notag\\
&\qquad{}+2u\biggl( \frac{\partial f}{\partial z}\frac{\partial \tilde{f}}{\partial u}-\frac{\partial f}{\partial u}\frac{\partial \tilde{f}}{\partial z}\biggr)+f\frac{\partial \tilde{f}}{\partial z}+\tilde{f}\frac{\partial f}{\partial z}
=\tilde{\rho},\label{e-6}
\\
&2f\frac{\partial \tilde{g}}{\partial u}+2\tilde{f}\frac{\partial g}{\partial v}+4z\biggl( \frac{\partial g}{\partial v}\frac{\partial \tilde{g}}{\partial u}-
\frac{\partial g}{\partial u}\frac{\partial \tilde{g}}{\partial v}\biggr) +2v\biggl( \frac{\partial g}{\partial v}\frac{\partial \tilde{g}}{\partial z}-\frac{\partial g}{\partial z}\frac{\partial \tilde{g}}{\partial v}\biggr) \notag \\
&\qquad{}+2u\biggl( \frac{\partial g}{\partial z}\frac{\partial \tilde{g}}{\partial u}-\frac{\partial g}{\partial u}\frac{\partial \tilde{g}}{\partial z}\biggr)-g\frac{\partial \tilde{g}}{\partial z}-\tilde{g}\frac{\partial g}{\partial z}
=-\tilde{\rho}.\label{e-7}
\end{align}
Note that comparing $\eqref{1.3}$ and $\eqref{1.4}$ with \eqref{2.18} and \eqref{2.19}, it follows that
\begin{equation}\label{2.27}
f=A\cos(\varphi)\sqrt{1-\frac{u}{A^{2}}+\varphi_{1}(z)}, \qquad g=B\sin(\varphi)\sqrt{1-\frac{v}{B^{2}}+\varphi_{2}(z)}
\end{equation}
and
\begin{equation}\label{2.28}
\tilde{f}=B\cos(\psi)\sqrt{1-\frac{v}{B^{2}}+\varphi_{2}(z)}, \qquad \tilde{g}=A\sin(\varphi)\sqrt{1-\frac{u}{A^{2}}+\varphi_{1}(z)}.
\end{equation}
We have checked that functions $f$, $g$, $\tilde {f}$, $\tilde {g}$ in \eqref{2.27} and \eqref{2.28} are solutions of the seven differential equations \eqref{e-1}--\eqref{e-7}. Also, it is easy to verify that the special solutions \eqref{1} and \eqref{1.1} satisfy the differential equations \eqref{e-1}--\eqref{e-7}.

In particular, for $\tilde{f}=\sqrt{1-v}$, $\tilde{g}=0$ and $g\equiv g(v)$, from equations \eqref{e-1}--\eqref{e-7} we find that
\begin{equation*}
f=\sqrt{\bigl( 1-\tilde{\rho}^{2}\bigr)\bigl( 1-u+z^{2}\bigr) }, \qquad g=\tilde{\rho}\sqrt{1-v},\qquad
h=\sqrt{\bigl( 1-\tilde{\rho}^{2}\bigr)\bigl( 1-u-v+uv+z^{2}-vz^{2}\bigr) }.
\end{equation*}

In the case, where $h=\tilde{h}$, we have
\begin{equation*}
AB\cos( \varphi +\psi) =1 \qquad \text{and} \qquad 1-\tilde{\rho}^{2}=\frac{1}{A^{2}B^{2}},
\end{equation*}
which implies that
\begin{equation*}
\tilde{h}=\sqrt{1-\tilde{\alpha}^{2}\tilde{X}^{2}-
	\tilde{\beta}^{2}\tilde{P}^{2}+2\tilde{\rho}\tilde{\alpha}\tilde{\beta}\tilde{X}\tilde{P}-\frac{1-\tilde{\rho}^{2}}{2}\tilde{\alpha}^{2}\tilde{\beta}^{2}M^{2}}.
\end{equation*}

In the case when $\tilde{\rho}=1$, we find that $\cos( \varphi +\psi) =0$, which implies $\tilde{h}=0$ and then from~\eqref{2} and~\eqref{2.1} we have $\tilde{X}_{\mu}=\frac{\tilde{\beta}}{\tilde{\alpha}}\tilde{P}_{\mu}$.
Similarly, in the case when $\tilde{\rho}=-1$, we have $\tilde{h}=0$ and $\tilde{X}_{\mu}=-\frac{\tilde{\beta}}{\tilde{\alpha}}\tilde{P}_{\mu}$.

\section{General solution}\label{sec-3}

\begin{Proposition}
If $\hat{\mathcal{L}}$ is an operator acting on the deformed phase space spanned by~$\hat{x}_{\mu}$ and~$\hat{p}_{\mu}$~as
\begin{gather*}
\hat{\mathcal{L}}\bigl( \hat{x}_{\mu}\bigr) =\biggl\lbrace \frac{1}{\alpha \beta}h, \hat{x}_{\mu}\biggr\rbrace=\frac{\beta}{\alpha}\hat{p}_{\mu},
\qquad
\hat{\mathcal{L}}\bigl( \hat{p}_{\mu}\bigr) =\biggl\lbrace \frac{1}{\alpha \beta}h, \hat{p}_{\mu}\biggr\rbrace=-\frac{\alpha}{\beta}\hat{x}_{\mu},
\end{gather*}
where $h$ is defined in \eqref{1.5} and \eqref{1.6},
then it holds
\begin{gather}\label{2.5}
\bigl( {\rm e}^{\varphi\hat{\mathcal{L}}}\bigr) \bigl( \hat{x}_{\mu}\bigr)=\hat{x}_{\mu}\cos{\varphi}+\frac{\beta}{\alpha}\hat{p}_{\mu}\sin{\varphi},
\\
\label{1.8}
\bigl( {\rm e}^{-\psi\hat{\mathcal{L}}}\bigr) \bigl( \hat{p}_{\mu}\bigr)=\hat{p}_{\mu}\cos{\psi}+\frac{\alpha}{\beta}\hat{x}_{\mu}\sin{\psi}.
\end{gather}
\end{Proposition}

\begin{proof}
First, we have that
\begin{gather}
\bigl( {\rm e}^{\varphi\hat{\mathcal{L}}}\bigr) \bigl( \hat{x}_{\mu}\bigr)=\hat{x}_{\mu}+ \frac{\varphi}{\alpha \beta}\bigl\lbrace h, \hat{x}_{\mu}\bigr\rbrace+\frac{1}{2!}\biggl( \frac{\varphi}{\alpha \beta}\biggr) ^{2}\bigl\lbrace h,\bigl\lbrace h, \hat{x}_{\mu}\bigr\rbrace \bigr\rbrace
+ \cdots \nonumber \\
\hphantom{\bigl( {\rm e}^{\varphi\hat{\mathcal{L}}}\bigr) \bigl( \hat{x}_{\mu}\bigr)=}{}
+ \frac{1}{n!}\biggl( \frac{\varphi}{\alpha \beta}\biggr) ^{n}
\underset{n \ \text{times}}{\underbrace{\bigl\lbrace h,\ldots,\bigl\lbrace h,\bigl\lbrace h,\hat{x}_{\mu}\bigr\rbrace\bigr\rbrace \ldots\bigr\rbrace }}+\cdots.\label{2.2}
\end{gather}
By induction on $n$ and using \eqref{1.7}, we prove the relations
\begin{gather}\label{2.3}
\underset{n=2k \ \text{times}}{\underbrace{\bigl\lbrace h,\ldots,\bigl\lbrace h,\bigl\lbrace h,\hat{x}_{\mu}\bigr\rbrace\bigr\rbrace \ldots\bigr\rbrace }}=( -1) ^{k}\alpha^{2k}\beta^{2k}\hat{x}_{\mu}, \qquad k=1,2,\ldots,
\\
\label{2.4}
\underset{n=2k+1 \ \text{times}}{\underbrace{\bigl\lbrace h,\ldots,\bigl\lbrace h,\bigl\lbrace h,\hat{x}_{\mu}\bigr\rbrace\bigr\rbrace \ldots\bigr\rbrace }}=( -1) ^{k}\alpha^{2k}\beta^{2k+2}\hat{p}_{\mu}, \qquad k=0,1,\ldots.
\end{gather}
For the case when $n=2k$, it is easy to see, using \eqref{1.7} that for $k=1$, we get
\begin{equation*}
\bigl\lbrace h,\bigl\lbrace h, \hat{x}_{\mu}\bigr\rbrace \bigr\rbrace =-\alpha^{2}\beta^{2}\hat{x}_{\mu}.
\end{equation*}
Assume that the relation \eqref{2.3} holds for some $k>1$. Then by the induction assumption and using \eqref{1.7}, we have
\begin{align*}
\underset{n=2k+2 \ \text{times}}{\underbrace{\bigl\lbrace h,\ldots,\bigl\lbrace h,\bigl\lbrace h,\hat{x}_{\mu}\bigr\rbrace\bigr\rbrace \ldots\bigr\rbrace }} =( -1) ^{k}\alpha^{2k}\beta^{2k}\bigl\lbrace h,\bigl\lbrace h, \hat{x}_{\mu}\bigr\rbrace \bigr\rbrace =( -1) ^{k+1}\alpha^{2( k+1) }\beta^{2( k+1) }\hat{x}_{\mu}.
\end{align*}
Similarly, by induction on $n$ and using \eqref{1.7}, we prove that relation \eqref{2.4} holds. Now, inserting relations \eqref{2.3} and \eqref{2.4} in \eqref{2.2}, we finally prove that \eqref{2.5} holds. Also, in a similar way, we prove~\eqref{1.8}.
\end{proof}

Let us now define an operator $\mathcal{L}_{G}$ acting on a canonical phase space spanned by $x_{\mu}$ and~$p_{\mu}$~as
\begin{gather*}
\mathcal{L}_{G} ( f ) = \lbrace G,f \rbrace,
\end{gather*}
where $G(x,p)$ and $f(x,p)$ are functions on classical phase space.
Furthermore, we define
\begin{gather*}
\hat{x}_{\mu}^{(0)}=\sqrt{1-\beta^{2}p^{2}}x_{\mu},
\qquad
\hat{p}_{\mu}^{(0)}=\sqrt{1-\alpha^{2}x^{2}}p_{\mu}
\end{gather*}
and then construct an operator $\mathcal{O}$ such that
\begin{equation}\label{2.6}
\bigl\lbrace \mathcal{O}\bigl( \hat{x}_{\mu}^{(0)}\bigr), \hat{p}_{\nu}^{(0)}\bigr\rbrace =\eta_{\mu\nu}h,
\end{equation}
where $h$ is given in \eqref{1.5}. The general structure of an operator $\mathcal{O}$ acting on $\hat{x}_{\mu}^{(0)}$ is
\begin{gather*}
\mathcal{O}\bigl( \hat{x}_{\mu}^{(0)}\bigr)=\bigl( {\rm e}^{\mathcal{L}_{G}}\bigr) \bigl( \hat{x}_{\mu}^{(0)}\bigr)=\hat{x}_{\mu}^{(0)}+ \bigl\lbrace G, \hat{x}_{\mu}^{(0)}\bigr\rbrace+\frac{1}{2!}\bigl\lbrace G,\bigl\lbrace G, \hat{x}_{\mu}^{(0)}\bigr\rbrace \bigr\rbrace \\
\hphantom{\mathcal{O}\bigl( \hat{x}_{\mu}^{(0)}\bigr)=\bigl( {\rm e}^{\mathcal{L}_{G}}\bigr) \bigl( \hat{x}_{\mu}^{(0)}\bigr)=}{}
+ \frac{1}{3!} \bigl\lbrace G\bigl\lbrace G,\bigl\lbrace G, \hat{x}_{\mu}^{(0)}\bigr\rbrace \bigr\rbrace \bigr\rbrace + \cdots	.
\end{gather*}
Solving perturbatively \eqref{2.6}, a unique solution was found for $G$ in \cite{Meljanac-4},
\begin{equation}\label{2.7}
G=\sum_{n=1}^{\infty}\alpha^{2n}\beta^{2n}g_{2n},
\end{equation}
where
\begin{equation*}
g_{2n}=\frac{ ( -1 ) ^{n}\cdot ( xp ) ^{2n+1}}{2n\cdot ( 2n+1 ) }.
\end{equation*}
The summation of the equation \eqref{2.7} for $G$ gives an exact result,
\begin{equation*}
G=\frac{1}{\alpha\beta}\biggl( z\biggl( 1-\frac{1}{2}\ln \bigl( 1+z^{2}\bigr) \biggr) -\arctan z\biggr).
\end{equation*}

If we fix $\tilde{P}_{\mu}^{(0)}=\hat{p}_{\mu}^{(0)}=\sqrt{1-\alpha^{2}x^{2}}p_{\mu}$, then we obtain the corresponding
\begin{gather*}
\tilde{X}_{\mu}^{(1)}=A\bigl( {\rm e}^{\varphi\hat{\mathcal{L}}}\circ
{\rm e}^{\mathcal{L}_{G}}\bigr) \bigl( \hat{x}_{\mu}^{(0)}\bigr) =A\bigl( {\rm e}^{\varphi\hat{\mathcal{L}}}\circ
	{\rm e}^{\mathcal{L}_{G}}\bigr) \bigl( \sqrt{1-\beta^{2}p^{2}}x_{\mu}\bigr),
\\
h^{(1,0)}=\sqrt{\left( 1-\beta^{2}p^{2}+z^{2}\right) \left( 1-\alpha^{2}x^{2}\right)}
\end{gather*}
with the property
\begin{equation*}
\bigl\lbrace \tilde{X}_{\mu}^{(1)}, \tilde{P}_{\mu}^{(0)}\bigr\rbrace =\eta_{\mu\nu}A\cos ( \varphi ) h^{(1,0)}+A \sin ( \varphi ) \alpha \beta M_{\mu\nu}.
\end{equation*}

The composition of mappings ${\rm e}^{\varphi\hat{\mathcal{L}}}\circ
{\rm e}^{\mathcal{L}_{G}}$ can be calculated perturbatively applying the BCH formula.
\begin{Proposition}
Application of the BCH formula to a Poisson algebra with the classical Poisson brackets gives
\begin{equation}\label{3}
{\rm e}^{\mathcal{L}_{A}}\circ
{\rm e}^{\mathcal{L}_{B}}={\rm e}^{\mathcal{L}_{C}},
\end{equation}
where $C=A+B+\frac{1}{2} \lbrace A,B \rbrace +\frac{1}{12} \lbrace A, \lbrace A,B \rbrace \rbrace -\frac{1}{12} \lbrace B, \lbrace A,B \rbrace \rbrace +\cdots $.
\end{Proposition}
\begin{proof}
From the BCH formula, we have
\begin{equation*}
{\rm e}^{\mathcal{L}_{A}}\circ
{\rm e}^{\mathcal{L}_{B}}={\rm e}^{ ( \mathcal{L}_{A}+\mathcal{L}_{B}+\frac{1}{2} [ \mathcal{L}_{A},\mathcal{L}_{B}
 ] +\frac{1}{12} [\mathcal {L}_{A}, [ \mathcal{L}_{A}, \mathcal{L}_{B} ]
 ]-\frac{1}{12} [ \mathcal{L}_{B}, [ \mathcal{L}_{A}, \mathcal{L}_{B} ] ] +\cdots
 ) },
\end{equation*}
which implies that in order to prove \eqref{3} is sufficient to prove that the relation
\begin{equation}\label{3.1}
 [ \mathcal{L}_{A_{1}},\ldots [ \mathcal{L}_{A_{n-1}}, [ \mathcal{L}_{A_{n}},\mathcal{L}_{B} ] ] \ldots ] =\mathcal{L}_ { ( \lbrace A_{1},\ldots \lbrace A_{n-1}, \lbrace A_{n},B \rbrace \rbrace \ldots \rbrace ) }
\end{equation}
holds.
We prove \eqref{3.1} by induction on $n$. For $n=1$, using the Jacobi identity, we get
\begin{align*}
 ( [ \mathcal{L}_{A_{1}},\mathcal{L}_{B}
 ] ) ( F(x,p) ) & = ( \mathcal{L}_{A_{1}}\mathcal{L}_{B}- \mathcal{L}_{B}\mathcal{L}_{A_{1}} ) ( F(x,p) ) = \lbrace A_{1}, \lbrace B,F \rbrace \rbrace - \lbrace B, \lbrace A_{1},F \rbrace \rbrace\\
& = \lbrace \lbrace A_{1}, B \rbrace ,F \rbrace =\mathcal{L}_{ ( \lbrace A_{1}, B \rbrace ) } ( F(x,p) ).
\end{align*}
Let us assume that the relation \eqref{3.1} holds for some $n>1$. Then by the induction assumption and using the Jacobi identity, we have
\begin{gather*}
( [ \mathcal{L}_{A_{1}},\ldots [ \mathcal{L}_{A_{n}}, [ \mathcal{L}_{A_{n+1}},\mathcal{L}_{B} ] ] \ldots ] ) ( F(x,p) )\\
\qquad{} = ( A_{1} [ \mathcal{L}_{A_{2}},\ldots [ \mathcal{L}_{A_{n}}, [ \mathcal{L}_{A_{n+1}},\mathcal{L}_{B} ] ] \ldots ]
 - [ \mathcal{L}_{A_{2}},\ldots [ \mathcal{L}_{A_{n}}, [ \mathcal{L}_{A_{n+1}},\mathcal{L}_{B} ] ] \ldots ]A_{1} ) ( F(x,p) )\\
\qquad{} = \lbrace A_{1}, \lbrace \lbrace A_{2},\ldots \lbrace A_{n}, \lbrace A_{n+1},B \rbrace \rbrace \ldots \rbrace, F \rbrace \rbrace\!
 -\! \lbrace \lbrace A_{2},\ldots \lbrace A_{n}, \lbrace A_{n+1},B \rbrace \rbrace \ldots \rbrace, \lbrace A_{1}, F \rbrace \rbrace \\
\qquad{} = \lbrace \lbrace A_{1}, \lbrace A_{2},\ldots \lbrace A_{n}, \lbrace A_{n+1},B \rbrace \rbrace \ldots \rbrace \rbrace ,F \rbrace
= \mathcal{L}_ { ( \lbrace A_{1},\ldots \lbrace A_{n}, \lbrace A_{n+1},B \rbrace \rbrace \ldots \rbrace ) } ( F(x,p) ).
\tag*{\qed}
\end{gather*}
\renewcommand{\qed}{}
\end{proof}

Let us denote ${\rm e}^{\mathcal{L}_{\tilde{G}}}={\rm e}^{\varphi\hat{\mathcal{L}}}\circ
{\rm e}^{\mathcal{L}_{G}}$. Then in the first order we get $\mathcal{L}_{\tilde{G}}=\varphi \hat{\mathcal{L}}+ \mathcal{L}_{G}+\frac{1}{2}\varphi \bigl\lbrace \hat{\mathcal{L}}, \mathcal{L}_{G}\bigr\rbrace \pm \cdots$. The most general realizations of $\tilde{X}_{\mu}$, $\tilde{P}_{\mu}$ and $h$ are obtained using the group of automorphisms applied to the special solution $\tilde{X}_{\mu}^{(1)}$, $\tilde{P}_{\mu}^{(0)}$ and $h^{(1,0)}$, namely $\tilde{X}_{\mu}=O_{F}\bigl( \tilde{X}_{\mu}^{(1)}\bigr)$, $\tilde{P}_{\mu}=O_{F}\bigl( \tilde{P}_{\mu}^{(0)}\bigr)$ and $h=O_{F}\bigl( h^{(1,0)}\bigr)$, where $O_{F} ={\rm e}^{\mathcal{L}_{F}}$ and $F$ is arbitrary function of $\alpha^{2}x^{2}$, $\beta^{2}p^{2}$ and $z$. Alternatively, we can write
\begin{gather*}
\tilde{X}_{\mu}=A\bigl( {\rm e}^{\varphi\hat{\mathcal{L}}}\circ
{\rm e}^{\mathcal{L}_{F}}\bigr) \bigl( \hat{x}_{\mu}^{(1)}\bigr) =A\bigl( {\rm e}^{\varphi\hat{\mathcal{L}}}\circ
{\rm e}^{\mathcal{L}_{F}}\bigr) \bigl( \sqrt{1-\beta^{2}p^{2}+z^{2}}x_{\mu}\bigr),
\\
\tilde{P}_{\mu}=B\bigl( {\rm e}^{-\psi\hat{\mathcal{L}}}\circ
{\rm e}^{\mathcal{L}_{F}}\bigr) \bigl( \hat{p}_{\mu}^{(0)}\bigr) =B\bigl( {\rm e}^{-\psi\hat{\mathcal{L}}}\circ
{\rm e}^{\mathcal{L}_{F}}\bigr) \bigl( \sqrt{1-\alpha^{2}x^{2}}p_{\mu}\bigr).
\end{gather*}
These solutions satisfy the above seven differential equations \eqref{e-1}--\eqref{e-7}.
\begin{Remark}
Note that the Poisson brackets are covariant under the action of ${\rm e}^{\mathcal{L}_{F}}$ i.e.~they satisfy
\begin{equation*}
 \bigl\lbrace \bigl( {\rm e}^{\mathcal{L}_{F}} \bigr)(f), \bigl( {\rm e}^{\mathcal{L}_{F}} \bigr)(g) \bigr\rbrace = \bigl( {\rm e}^{\mathcal{L}_{F}} \bigr) ( \lbrace f, g \rbrace )
\end{equation*}
for any function $F$, $f$, $g$.
Also $ \bigl( {\rm e}^{\mathcal{L}_{F}}\bigr)\bigl( M_{\mu\nu}\bigr) =M_{\mu\nu}$ if $\bigl\lbrace F, M_{\mu\nu} \bigr\rbrace =0$ and $ \bigl( {\rm e}^{\mathcal{L}_{F}}\bigr)\bigl( \eta_{\mu\nu}\bigr) =\eta_{\mu\nu}$.
\end{Remark}

\begin{Remark}
Using realizations of the Snyder model \cite{Martinic-2}, we can
write the corresponding realizations of the Snyder Poisson model. From the results that were found in \cite{Martinic-2}, we get
\begin{gather*}
\hat{x}_{\mu}= \bigl( {\rm e}^{\mathcal{L}_{G}} \bigr) \bigl( x_{\mu}+\beta^{2}(x\cdot p)p_{\mu} \bigr)
=x_{\mu}\varphi_{1}(u)+\beta^{2}(x\cdot p)p_{\mu}\varphi_{2}(u),
\\
\hat{p}_{\mu}=\bigl( {\rm e}^{\mathcal{L}_{G}}\bigr)\bigl( p_{\mu}\bigr) =p_{\mu}\frac{1}{\varphi_{1}(u)},
\end{gather*}
where
\begin{gather*}
G=(x\cdot p)F(u), \qquad \varphi_{2}(u)=\frac{1+\dot{\varphi}_{1}(u)\varphi_{1}(u)}{\varphi_{1}(u)-2u\dot{\varphi}_{1}(u)}, \qquad \dot{\varphi}_{1}=\frac{{\rm d}\varphi_{1}(u)}{{\rm d}u} \qquad \text{and} \qquad u=\beta^{2}p^{2}.
\end{gather*}
\end{Remark}

\section{Related Poisson models}\label{sec-4}

In this section, we introduce related Poisson algebras generalizing \eqref{0.3}, \eqref{0.4} and \eqref{1.7}:
\begin{gather*}
\bigl\lbrace \hat{x}_{\nu},\hat{x}_{\nu}\bigr\rbrace =\epsilon_{1}\beta^{2}M_{\mu\nu}, \qquad \bigl\lbrace \hat{p}_{\nu},\hat{p}_{\nu}\bigr\rbrace =\epsilon_{2}\alpha^{2}M_{\mu\nu},
\qquad
\bigl\lbrace \hat{x}_{\mu}, \hat{p}_{\nu}\bigr\rbrace =\eta_{\mu\nu}h,
\\
\bigl\lbrace h, \hat{x}_{\mu}\bigr\rbrace =\epsilon_{1}\beta^{2}\hat{p}_{\mu}, \qquad \bigl\lbrace h, \hat{p}_{\mu}\bigr\rbrace =-\epsilon_{2}\alpha^{2}\hat{x}_{\mu},
\end{gather*}
where $\epsilon_{1}^{2}=\epsilon_{2}^{2}=1$. The Yang Poisson model in
Section \ref{sec-1} corresponds to the $\epsilon_{1}=\epsilon_{2}=1$ case which is related to $\mathfrak{o}(1,5)$ algebra.

Now, let us consider the \textit{case when $\epsilon_{1}=\epsilon_{2}=-1$}.
A special solution for realizations of~$\hat{x}_{\mu}$ and~$\hat{p}_{\mu}$ in this case is given as
\begin{gather*}
\hat{x}_{\mu}=x_{\mu}\sqrt{1+\beta^{2}p^{2}+\phi_{1}(z)},
\qquad
\hat{p}_{\mu}=p_{\mu}\sqrt{1+\alpha^{2}x^{2}+\phi_{2}(z)},
\end{gather*}
where
\begin{gather*}
\phi_{1}\phi_{2}+\phi_{1}+\phi_{2}=z^{2}, \qquad z=\alpha\beta ( xp),
\qquad
h=\sqrt{\bigl( 1+\beta^{2}p^{2}+\phi_{1}(z)\bigr) \bigl( 1+\alpha^{2}x^{2}+\phi_{2}(z)\bigr) }.
\end{gather*}
The new generators $\tilde{X}_{\mu}$ and $\tilde{P}_{\mu}$, linear in $\hat{x}_{\mu}$, $\hat{p}_{\mu}$ and $M_{\mu\nu}$, are given in \eqref{x-2} and \eqref{p-2}. Also, a~new algebra, related to $\mathfrak{o}(3,3)$ algebra, generated by $\tilde{X}_{\mu}$, $\tilde{P}_{\mu}$,
$M_{\mu\nu}$ and $\tilde{h}$ is given in \eqref{1.10}--\eqref{1.16}, where
\begin{equation*}
\tilde{\rho}=-\sin(\varphi +\psi)+\frac{ab}{AB}, \qquad \tilde{A}=-A^{2}+a^{2} \qquad \text{and} \qquad \tilde{B}=-B^{2}+b^{2}.
\end{equation*}

In the \textit{case when $\epsilon_{1}=1$ and $\epsilon_{2}=-1$}, a special solution is given by
\begin{gather*}
\hat{x}_{\mu}=x_{\mu}\sqrt{1-\beta^{2}p^{2}+\phi_{1}(z)},
\qquad
\hat{p}_{\mu}=p_{\mu}\sqrt{1+\alpha^{2}x^{2}+\phi_{2}(z)},
\end{gather*}
where
\begin{equation*}
\phi_{1}\phi_{2}+\phi_{1}+\phi_{2}=-z^{2}, \qquad z=\alpha\beta ( xp),
\qquad
h=\sqrt{\bigl( 1-\beta^{2}p^{2}+\phi_{1}(z)\bigr) \bigl( 1+\alpha^{2}x^{2}+\phi_{2}(z)\bigr) }.
\end{equation*}
The new generators $\tilde{X}_{\mu}$ and $\tilde{P}_{\mu}$ are given by
\begin{gather}\label{x-4}
\tilde{X}_{\mu}=A\biggl( \cosh \varphi \hat{x}_{\mu}+\frac{\beta}{\alpha}\sinh \varphi \hat{p}_{\mu}\biggr) +\beta a_{\nu}M_{\mu\nu},
\\ \label{p-4}
\tilde{P}_{\mu}=B\biggl( \cosh \psi \hat{p}_{\mu}+\frac{\alpha}{\beta}\sinh \psi \hat{x}_{\mu}\biggr) +\alpha b_{\nu}M_{\mu\nu},
\end{gather}
where $M_{\mu\nu}=x_{\mu}p_{\nu}-x_{\nu}p_{\mu}$ and $A$, $B$, $\varphi$, $\psi$, $a_{\mu}$, $b_{\mu}$ are dimensionless parameters with $AB\neq 0$. New generators $\tilde{X}_{\mu}$ and $\tilde{P}_{\mu}$ together with
$ M_{\mu\nu}$ and $\tilde{h}$ generate a new algebra, isomorphic to the initial Yang Poisson algebra, which is given in \eqref{1.10}--\eqref{2.1}, where
\begin{gather*}
\tilde{h}=AB\cosh(\psi-\varphi)h+\beta a\tilde{P}-\alpha b \tilde{X}-\alpha\beta a_{\mu}b_{\nu}M_{\mu\nu},
\\
\tilde{\rho}=\sinh(\psi-\varphi)+\frac{ab}{AB}, \qquad \tilde{A}=A^{2}+a^{2} \qquad \text{and} \qquad \tilde{B}=-B^{2}+b^{2}.
\end{gather*}

Finally, in the last \textit{case when $\epsilon_{1}=-1$ and $\epsilon_{2}=1$} a special solution is given by
\begin{gather*}
\hat{x}_{\mu}=x_{\mu}\sqrt{1+\beta^{2}p^{2}+\phi_{1}(z)},
\qquad
\hat{p}_{\mu}=p_{\mu}\sqrt{1-\alpha^{2}x^{2}+\phi_{2}(z)},
\end{gather*}
where
\begin{equation*}
\phi_{1}\phi_{2}+\phi_{1}+\phi_{2}=-z^{2}, \qquad z=\alpha\beta ( xp),\qquad
h=\sqrt{\bigl( 1+\beta^{2}p^{2}+\phi_{1}(z)\bigr) \bigl( 1-\alpha^{2}x^{2}+\phi_{2}(z)\bigr) }.
\end{equation*}
The new generators $\tilde{X}_{\mu}$ and $\tilde{P}_{\mu}$ are given in
\eqref{x-4} and \eqref{p-4} and a new algebra generated by them together with $M_{\mu\nu}$ and $\tilde{h}$
is given in \eqref{1.10}--\eqref{2.1}, where
\begin{gather*}
\tilde{h}=AB\cosh(\psi-\varphi)h+\beta a\tilde{P}-\alpha b \tilde{X}-\alpha\beta a_{\mu}b_{\nu}M_{\mu\nu},
\\
\tilde{\rho}=-\sinh(\psi-\varphi)+\frac{ab}{AB}, \qquad \tilde{A}=-A^{2}+a^{2} \qquad \text{and} \qquad \tilde{B}=B^{2}+b^{2}.
\end{gather*}

Note that the algebras generated with cases when $\epsilon_{1}=1$, $\epsilon_{2}=-1$ and $\epsilon_{1}=-1$, $\epsilon_{2}=1$ are related to $\mathfrak{o}(2,4)$ algebra.

\section{Discussion}

The generalized Yang models are examples of noncommutative geometry on a background spacetime of constant curvature that display a duality between position and momentum manifolds.

In this paper, we
have obtained the exact realizations of a generalized Yang Poisson algebra on a canonical phase space related to the $\mathfrak{o}(1,5)$ algebra. These realizations are simpler than in the quantum case. The results we have obtained can be considered as a limit of the quantum formalism for $\hbar \rightarrow 0,$ presenting a classical approximation of the quantum realizations.

Moreover, we have discussed the Poisson algebras related to $\mathfrak{o}(3,3)$ and $\mathfrak{o}(2,4)$ algebras. These models correspond to different physical settings, namely, the case $\epsilon_2=-1$ is related to the symmetries
of anti-de Sitter spacetime, while $\epsilon_1=-1$ corresponds to the so-called anti-Snyder algebra \cite{Mignemi-2}. We recall that the anti-Snyder algebra enjoys rather different properties from the Snyder algebra, in particular concerning the existence of a maximum allowable momentum.

Possible applications of our results are in cosmology, since the present model could be useful in describing effects of noncommutativity in early stages of inflation, and in the investigation of modifications of the dynamics of simple mechanical systems caused by the deformed symplectic structure.
The most elementary example is the harmonic oscillator, that has been studied in~\cite{Meljanac-4} in the $\mathfrak{o}(1,5)$ case, exhibiting a modification with respect to the canonical solution, with the period that becomes energy dependent.
 In the $\mathfrak{o}(3,3)$ and $\mathfrak{o}(2,4)$ cases, one expect similar modifications, analogous to those found in \cite{Mignemi-3} for the related TSR setting.

In the present paper, a canonical phase space is defined with coordinates~$x_{\mu}$, momenta $p_{\mu}$ and Lorentz generators $M_{\mu \nu} = x_{\mu} p_{\nu} - x_{\nu} p_{\mu}$. However, it is possible to define the extended coordinates~$x_{\mu \nu}$ with corresponding momenta $p_{\mu \nu}$ and to interpret the Lorentz generators $M_{\mu \nu}$ as extended coordinates~$x_{\mu \nu}$ and $\tilde{h}$ as an additional scalar coordinate. In this framework, the algebra \eqref{1.10}--\eqref{1.17} is generated by 15 coordinates. The corresponding realizations can be obtained from the quantum realizations of the generalized Yang models, presented in \cite{Lukierski, Martinic-3}.

\subsection*{Acknowledgement}

S.~Mignemi thanks Gruppo Nazionale di Fisica Matematica for support.

\pdfbookmark[1]{References}{ref}
\LastPageEnding

\end{document}